\documentclass[11pt]{article}
\usepackage[T1]{fontenc}
\usepackage{amsmath,amsxtra,amssymb,latexsym,amscd,amsthm,pb-diagram,fancyhdr,euscript}
\usepackage{amsthm}
\usepackage{indentfirst}
\usepackage{epsfig}
\usepackage{mathrsfs}
\usepackage{yfonts}
\usepackage{dsfont}
\usepackage{slashed}
\usepackage{tkz-tab}
\usepackage{hyperref}
\hypersetup{
   colorlinks,
    citecolor=blue,
    filecolor=black,
    linkcolor=red,
    urlcolor=magenta
}
\hypersetup{linktocpage}

\newcommand{\C}{\mathbb{C}} 
 
\newcommand{\R}{\mathbb{R}} \newcommand{\N}{\mathbb{N}}
\renewcommand{\S}{\mathbb{S}}
\newcommand{\s}{\mathrm{S}}

\newcommand{\Z}{\mathbb{Z}}
\newcommand{\HB}{\mathbb{H}}

\newcommand{\I}{{\cal I}_{1/2}}

\renewcommand{\d}{\mathrm{d}}

\newcommand{\scrh}{{\mathscr H}}
\newcommand{\scrs}{{\mathscr S}}
\newcommand{\tilt}{\tilde{t}}

\newcommand{\hlm}{{\scrh}^L_{r_{_-}}}
\newcommand{\hrm}{{\scrh}^R_{r_{_-}}}
\newcommand{\hlp}{{\scrh}^L_{r_{_+}}}
\newcommand{\hrp}{{\scrh}^R_{r_{_+}}}

\newcommand{\sld}{\slashed{D}}

\newtheorem{definition}{Definition}[section]
\newtheorem{theorem}{Theorem}[section]
\newtheorem{proposition}{Proposition}[section]

\newtheorem{lemma}{Lemma}[section]
\newtheorem{remark}{Remark}[section]
\begin{document}
\mbox{} \thispagestyle{empty}

\begin{center}
\bf{\huge Scattering theory for Dirac fields inside a Reissner-Nordström-type black hole}

\vspace{0.1in}

{Dietrich HÄFNER$^{\textcolor{red}{1}}$, Mokdad MOKDAD\footnote{Université Grenoble Alpes, Institut Fourier, 100 rue des Maths, 38610 Gières, France \\
E-mails: Dietrich.Hafner@univ-grenoble-alpes.fr, Mokdad.Mokdad@univ-grenoble-alpes.fr.} \& Jean-Philippe NICOLAS\footnote{LMBA, UMR CNRS 6205, Université de Brest, 6 avenue Victor Le Gorgeu, 29200 Brest, France \\
E-mail: Jean-Philippe.Nicolas@univ-brest.fr.}}
\end{center}

\vspace{0.5cm}

{\bf Abstract.}  We show asymptotic completeness for the massive charged Dirac equation between the black hole and Cauchy horizons of a sub-extremal ((Anti-) De Sitter) Reissner-Nordstr\"om black hole. 

\vspace{0.5cm}


\vspace{0.1in}


\tableofcontents

\section{Introduction}

The aim of this paper is to show how time dependent methods apply to scattering theory in black hole interiors. Scattering theory has now a long tradition in General Relativity, but concerns mainly field equations outside black holes, see \cite{DK}, \cite{Ba}, \cite{HN}, \cite{Ni2016}, \cite{GGH}, \cite{DRSR}, \cite{Mo2} and references therein for an overview. Time decaying situations are in most cases easier than time independent situations. Whereas the decay in time of the hamiltonian towards some simplified hamiltonian is in time dependent situations clear right from the beginning, in time independent situations it must be established via so called propagation estimates and holds only along the evolution. We consider in this paper the massive charged Dirac equation on the ((Anti)-De Sitter-)Reissner-Norstr\"om metric in the sub-extremal case. The Dirac equation has the advantage that it possesses a conserved positive quantity which makes the application of time dependent methods particularly convenient. Our main result is an asymptotic completeness result between the black hole horizon and the Cauchy horizon. The charge entails some additional technical difficulties as we have to work in a different gauge near each horizon. Our result is formulated in terms of wave operators as well as in terms of existence and uniqueness of solutions of a characteristic Cauchy problem. 

Asymptotic completeness for the wave equation inside a Reissner-Nordstr\"om black hole was established by C. Kehle and Y. Shlapentokh-Rothman in \cite{KeShla} using Fourier analysis. They build up on earlier work by
S. Chandrasekhar and J.B. Hartle \cite{ChaHa}. One of the main motivations of these works is to show the instability of the Cauchy horizon. The lack of transverse regularity of the scalar field at the Cauchy horizon is established in both \cite{ChaHa} and \cite{KeShla}. Blue shift instability of the field near the Cauchy horizon was first observed numerically in the founding work of  R. Penrose and M. Simpson \cite{PeSi}. In order to establish the lack of transverse regularity, it seems essential to have a spectral decomposition of the scattering matrix, which, in our case, could be done a posteriori. These questions are left for future work.

There is no timelike Killing vector between the black hole horizon and the Cauchy horizon and the situation is therefore time dependent. Asymptotic completeness can then be established for the Dirac equation essentially by Cook's method using the conservation of the $L^2$ norm. The Killing vector fields commute with the equation and span a spacelike distribution that is integrable and forms the foliation by the Cauchy hypersurfaces of constant $r$. As a consequence, the different wave and scattering operators preserve tangent regularity.

The paper is organised as follows. In Section \ref{GeometricSetting}, we present a general geometric setting that covers the interiors of sub-extremal Reissner-Nordström black holes as well as their de Sitter and anti-de Sitter analogues. Section \ref{DiracEquation} describes the charged Dirac equation with its general properties (gauge freedom in the charged case, conserved current), we derive an explicit expression using the Newman-Penrose formalism and we make some simplifications to the equation by incorporating the $3$-volume density into the spinor. Then we define the evolution system and recall its essential features. Section \ref{ScatteringTheory} is devoted to the scattering theory and its geometrical interpretation is given in Section \ref{GeometricalInterpretation}. In an appendix, we derive the conditions on the physical parameters of the anti-de Sitter-Reissner-Nordström metric that correspond to the different configurations of horizons, among which, the sub-extremal family.

Note that the present scattering theory could be constructed in an alternative geometrical manner using the approach known as conformal scattering (see \cite{MaNi2004,Ni2016,Mo2}), although no conformal rescaling is necessary here. This consists in pushing further what is done in the geometrical interpretation of Section \ref{GeometricalInterpretation}. Once the trace operators are constructed, we could prove directly their invertibility by solving the Goursat problem at the bifurcate horizons. This can be done following a similar approach to that in M. Mokdad \cite{Mo2} and will be the object of a subsequent paper.

\vspace{0.1in}

\noindent{\bf Acknowledgements.} D. Häfner and J.-P. Nicolas acknowledge support from the ANR funding ANR-16-CE40-0012-01. 
\vspace{0.1cm}

\noindent{\bf Notations and conventions.}
\begin{itemize}
\item Throughout the paper, we use the formalisms of abstract indices, $2$-component spinors and Newman-Penrose. Recall that a Newman-Penrose tetrad on a Lorentzian $4$-manifold $(\mathcal{M},g)$ is a set of four null vector fields $\{ l^a,n^a,m^a,\bar{m}^a\}$ that are a local basis of the complexified tangent bundle of $\mathcal{M}$. It is said to be normalized if
\[ l_an^a = -m_a\bar{m}^a = 1 \, .\]
If we assume that $\mathcal{M}$ is globally hyperbolic, then it admits a spin-structure and we denote by $\S^A$ and $\S^{A'}$ the bundles of left and right spinors over $\mathcal{M}$ and by $\S_A$ and $\S_{A'}$ their respective dual bundles. They satisfy
\[ \S^A \otimes \S^{A'} = T^a\mathcal{M} \otimes \C \]
and are endowed with symplectic forms $\varepsilon_{AB}$ and $\varepsilon_{A'B'}$ such that $g_{ab} = \varepsilon_{AB} \varepsilon_{A'B'}$. The bundle of Dirac spinors is given by $\S_A \oplus \S^{A'}$.

A spin-dyad, is a local basis $\{ o^A , \iota^A \}$ of $\S^A$. It is normalized if $o_A \iota^A =1$ and it is then called a spin-frame. The components of a spinor field in a spin-frame $\{ o^A , \iota^A\}$ are given as follows (see the part on spinor bases in Section 2.5 of \cite{PeRi} Vol. I): for a spinor field $\phi_A \in \Gamma (\S_A )$ (where $\Gamma (\S_A )$ denotes the space of sections of the spin-bundle $\S_A$), we have
\[ \phi_0 = \phi_A o^A \, ,~ \phi_1 = \phi_A \iota^A \]
and for $\chi^{A'} \in \Gamma (\S^{A'})$,
\[ \chi^{0'} = -\bar{\iota}_{A'} \chi^{A'} \, ,~ \chi^{1'} = \bar{o}_{A'} \chi^{A'} \, . \]
To a normalised Newman-Penrose tetrad, there corresponds a (unique modulo an overall sign) spin-frame such that
\begin{equation} \label{NPTSF}
l^a = o^A \bar{o}^{A'} \, ,~ n^a = \iota^A \bar{\iota}^{A'} \, ,~ m^a = o^A \bar{\iota}^{A'} \, ,~ \bar{m}^a = \iota^A \bar{o}^{A'} \, .
\end{equation}
\item Given $\cal M$ a smooth manifold, $\cal O$ an open set of $\cal M$ and $F$ a fiber bundle over $\mathcal{M}$, we denote by ${\cal C}^\infty_0 ({\cal O}\, ;~ F)$ the space of smooth sections of $F$ with compact support in ${\cal O}$.
\item Given a coordinate system and $x$ one of the variables, we shall denote
\[ D_x := \frac{1}{i} \frac{\partial}{\partial x}\, .\]
\end{itemize}

\section{Geometric setting} \label{GeometricSetting}

We work on a spacetime ${\cal M} = \R_t \times \R_x \times \s^2_\omega$ equipped with the metric
\begin{equation} \label{Metric}
g = f(r) \left( \d t^2 - \d x^2 \right) - r^2 \d \omega^2
\end{equation}
where $\d\omega^2$ is the euclidean metric on the unit $2$-sphere and $r$ is an implicit function of $t$ alone defined by the following requirements:
\begin{description}
\item[(H1)] the smooth function $f$ is defined on an interval $[r_- , r_+]$, $0<r_-<r_+<+\infty$, is positive on $]r_-,r_+[$ and has simple zeros at $r_\pm$;
\item[(H2)] $\frac{\d t }{\d r} = -\frac{1}{f(r)}$ with $t=0$ for an arbitrarily chosen value of $r$ in $]r_-,r_+[$, say $r_m=(r_-+r_+)/2$.
\end{description}
This implies that $r$ is a smooth, strictly decreasing function of $t$ on $\R$ and $r\rightarrow r_\mp$ as $t\rightarrow \pm \infty$. Moreover, we can obtain an equivalent of $r-r_\mp$ as $r\rightarrow r_\mp$. Let us do this for $r\rightarrow r_-$. We can write $f$ as
\[ f(r) = q(r) (r-r_-)\]
with $q \in {\cal C}^\infty ([r_-,r_+])$. Therefore,
\[ \frac{1}{f(r)}=\frac{c}{r-r_-} + \frac{d(r)}{q(r)} \]
where
\[ c = \frac{1}{q(r_-)} = \frac{1}{f'(r_-)} \, ,~ d(r) = \frac{1-cq(r)}{r-r_-} \in {\cal C}^\infty ([r_-,r_+]) \, ,\]
$f'$ denoting the derivative of $f$ with respect to $r$. It follows that
\[ t = -c \log (r-r_-) - \int_{r_m}^r \frac{d(s)}{q(s)} \d s + c \log (r_m-r_-) \, ,\]
whence
\[ r-r_- = (r_m-r_-)e^{-t/c}e^{-\frac{1}{c} \int_{r_m}^r \frac{d(s)}{q(s)} \d s} \, .\]
We therefore have
\begin{equation} \label{requivrpm}
r-r_\pm= \mp e^{2 \kappa_\pm t + w_\pm(r)}  \mbox{ as } t \rightarrow \mp \infty \, ,
\end{equation}
where $w_\pm \in {\cal C}^\infty ([r_-,r_+])$ and
\[ \kappa_\pm = -\frac12 f'(r_\pm) \]
are the surfaces gravities at the inner and outer horizons. Note that $\kappa_- <0$ and $\kappa_+ >0$. 

We denote by $\Sigma_t$ the level hypersurfaces of the time function $t$, i.e.
\[ \Sigma_t = \{ t \} \times \Sigma \, ,~ \Sigma = \R_x \times \s^2_\omega \, .\]
We choose the time orientation given by the timelike vector field $\frac{\partial}{\partial t}$.
 
Introducing the Eddington-Finkelstein variables
\begin{equation} \label{EFV}
u = t-x \, ,~ v = t+x \, ,
\end{equation}
the metric can be written in terms of the coordinates systems $(u,r,\omega)$ and $(v,r,\omega)$ and it has the same expression in both:
\begin{eqnarray}
g &=& - f(r) \d u^2 - 2 \d u \d r - r^2 \d \omega^2 \, , \label{REFMet}\\
 &=& - f(r) \d v^2 - 2 \d v \d r - r^2 \d \omega^2 \, . \label{AEFMet}
\end{eqnarray} 
Another useful form of the metric is in the coordinate system $(u,v,\omega )$
\begin{equation} \label{uvMet}
g = f(r) \d u \d v - r^2 \d \omega^2 \, .
\end{equation}
The expressions \eqref{REFMet} and \eqref{AEFMet} show that $g$ extends as a smooth non-degenerate metric to $[r_-,r_+]_r \times \R_u \times \s^2_\omega$ and to $[r_-,r_+]_r \times \R_v \times \s^2_\omega$, allowing us to glue the regions $\{ r=r_\pm\}$ each as a pair of smooth null boundaries, referred to respectively as the left and right inner and outer horizons:
\begin{eqnarray*}
\hlm &:=& \{ r=r_-\} \times \R_v \times \s^2_\omega \, ,\\
\hrm &:=& \{ r=r_-\} \times \R_u \times \s^2_\omega \, ,\\
\hlp &:=& \{ r=r_+\} \times \R_u \times \s^2_\omega \, ,\\
\hrp &:=& \{ r=r_+\} \times \R_v \times \s^2_\omega \, .
\end{eqnarray*}
The two components of each horizon meet in a smooth $2$-sphere, which we denote respectively by $\scrs_{r_+}$ for the outer horizon and $\scrs_{r_-}$ for the inner horizon. This can be seen using Kruskal-Szekeres-type coordinates. For the construction of $\scrs_{r_-}$, we put (recall that $\kappa_- <0$)
\[ T = -\frac{1}{2} e^{\kappa_- t} \left( e^{\kappa_-x} +
e^{-\kappa_-x} \right) \, , ~ X = -\frac{1}{2} e^{\kappa_-t}
\left( e^{\kappa_-x} - e^{-\kappa_-x} \right) \, . \]
In these coordinates, the metric $g$ reads
\[ g = \frac{f(r) e^{-2\kappa_-  t}}{\kappa_-^2} (\d T^2 - \d X^2) - r^2 \d \omega^2 \]
which, by \eqref{requivrpm}, extends smoothly and non degenerately to $\{ T=X \leq 0 \} \cup \{ T=-X \leq 0\}$ and restricts to $\{ T=X=0 \}$ as $-r_-^2 \d \omega^2$. The inner crossing sphere is then defined as
\[ \scrs_{r_-} = \{(0,0)\}_{(T,X)} \times \s^2_\omega \, .\]
It is reached as $v \rightarrow +\infty$ along $\hlm$, as $u \rightarrow +\infty$ along $\hrm$ and going towards the future along all lines of fixed $x$ and $\omega$. A similar choice of variables can be introduced to define the outer crossing sphere $\scrs_{r_+}$. See Figure \ref{PenroseDiagram} for a Penrose diagram of the spacetime $\mathcal{M}$.
\begin{figure}
\centering
\includegraphics[scale=1.1]{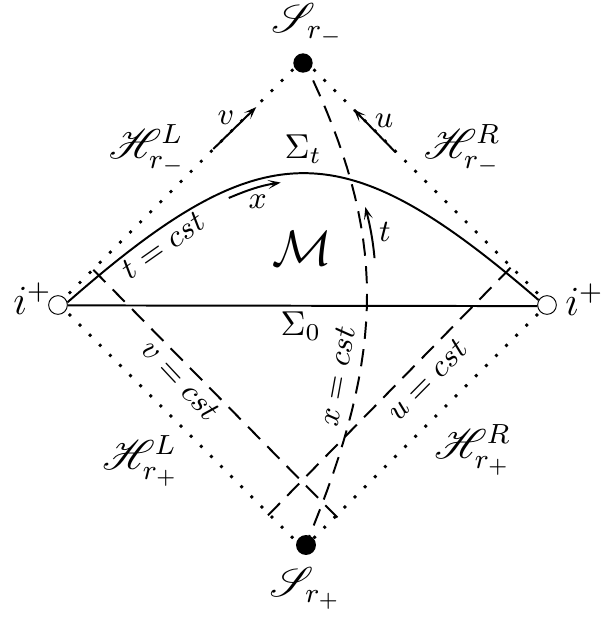}
\caption{\emph{A Penrose diagram of $\mathcal{M}$.}}
\label{PenroseDiagram}
\end{figure}
The future boundary of $\mathcal{M}$ is the set
\[ \{ r=r_- \} = \hlm \cup \scrs_{r_-} \cup \hrm \]
and its past boundary is
\[ \{ r=r_+ \} = \hlp \cup \scrs_{r_+} \cup \hrp \, . \]
The boundary of $\mathcal{M}$ misses two ``points'', denoted $i^+$ on the picture, by reference to future timelike infinity in the three physical cases on which this model is based: the sub-extremal Reissner-Nordström black hole and its de Sitter and anti-de Sitter generalisations. In all three cases, the metric is given on $\R_{\tilde{t}} \times ]0,+\infty [_r \times \s^2_\omega$ by
\begin{equation} \label{DSRNMet}
g = F \d \tilt^2 - \frac{1}{F} \d r^2 - r^2 \d \omega^2 \, ,~ F = 1 - \frac{2M}{r} + \frac{Q^2}{r^2} - \Lambda r^2 \, ,
\end{equation}
where $M$ and $Q$ are the mass and charge of the black hole and $\Lambda \in \R$ is the cosmological constant. We place ourselves between two consecutive simple horizons, corresponding to $0<r_-<r_+$, such that their respective surface gravities satisfy $\kappa_- < 0 < \kappa_+$. For $\Lambda = 0$ (Reissner-Nordström), this requires $0<\vert Q \vert <M$ and we have $r_\pm = M \pm \sqrt{M^2-Q^2}$. For $\Lambda >0$ (de Sitter-Reissner-Nordström), the conditions for having three simple horizons or one double and one simple horizon are detailed in M. Mokdad \cite{Mo1} (note that there is a similar study for de Sitter-Kerr black holes due to J. Borthwick \cite{Bo}). For $\Lambda <0$ (anti-de Sitter-Reissner-Nordström), the conditions are simpler and their derivation is the object of the appendix.

Our spacetime is $(\mathcal{M}, g)$, $\mathcal{M} = \R_{\tilde{t}} \times ]r_- , r_+ [_r \times \s^2_\omega$. The function $F$ is negative on $]r_-,r_+[$ and we introduce the Regge-Wheeler variable $r_*$ defined by
\[ \frac{\d r_*}{\d r} = \frac{1}{F} \mbox{ and } r_* =0 \mbox{ for, say, } r= \frac{r_-+r_+}{2} \, .\]
Then $r_*$ is an analytic diffeomorphism from $]r_-,r_+[$ onto $\R$ and $r \rightarrow r_\mp$ corresponds to $r_* \rightarrow \pm \infty$. Therefore, ${\mathcal M}$ can be described as $\R_{\tilt} \times \R_{r_*} \times \s^2_\omega$ with the metric
\[ g = F \left( \d \tilt^2 - \d r_*^2 \right) - r^2 \d \omega^2 \, . \]
Putting $t=r_*$, $x={\tilt}$, $f(r) = -F(r)$, we obtain the metric \eqref{Metric} and $f$ satisfies {\it (H1)-(H2)}.

\section{The charged massive Dirac equation} \label{DiracEquation}

The charged massive Dirac equation reads
\begin{equation} \label{DirEq}
\begin{cases} {\big( \nabla^{AA'} - iq A^{AA'} \big) \phi_A } & = {\frac{m}{\sqrt{2}} \chi^{A'} \, ,} \\ {\big( \nabla_{AA'} - i q A_{AA'} \big) \chi^{A'}} &= - {\frac{m}{\sqrt{2}} \phi_A \, ,} \end{cases}
\end{equation}
where
\begin{equation} \label{EMPot}
A_a \d x^a = \frac{Q}{r} \d x
\end{equation}
is the electrostatic potential of the ambiant electromagnetic field of the spacetime.

Equation \eqref{DirEq} is a hyperbolic equation and we work on a globally hyperbolic spacetime. Therefore, Leray's theorems \cite{Le} imply:
\begin{lemma} \label{LerayCP}
For any $s\in \R$, for any smooth compactly supported data $(\alpha_A \, ,~ \beta^{A'})$ on $\Sigma_s$, \eqref{DirEq} admits a unique solution $(\phi_A , \chi^{A'}) \in {\cal C}^\infty ({\mathcal M}\, ;~  \S_A \oplus \S^{A'} )$ such that
\[ \phi_A \vert_{t=s} = \alpha_A \, ,~ \chi^{A'} \vert_{t=s} = \beta^{A'} \, .\]
\end{lemma}

\subsection{Gauge freedom}

Equation \eqref{DirEq} has a natural gauge freedom that can be described as follows. Let $p$ be a smooth scalar function on $\mathcal{M}$, we put
\begin{equation}\label{gauge}
\tilde{A} = A + \d p \, ,~ (\tilde{\phi}_A , \tilde{\chi}^{A'} ) = e^{iqp} (\phi_A , \chi^{A'} )\, ,
\end{equation}
then $(\phi_A , \chi^{A'} )$ is a solution to the charged massive Dirac equation \eqref{DirEq} if and only if $(\tilde{\phi}_A , \tilde{\chi}^{A'} )$ satisfies
\begin{equation} \label{GTDirEq}
\begin{cases} {\big( \nabla^{AA'} - iq \tilde{A}^{AA'} \big) \tilde{\phi}_A } & = {\frac{m}{\sqrt{2}} \tilde{\chi}^{A'} \, ,} \\ {\big( \nabla_{AA'} - i q \tilde{A}_{AA'} \big) \tilde{\chi}^{A'}} &= - {\frac{m}{\sqrt{2}} \tilde{\phi}_A \, .} \end{cases}
\end{equation}
Indeed we have
\begin{eqnarray*}
\big( \nabla^{AA'} - iq \tilde{A}^{AA'} \big) \tilde{\phi}_A &=& e^{iqp} \big(\nabla^{AA'} \phi_A + iq (\nabla^{AA'} p) \phi_A - iq A^{AA'} \phi_A - i q (\nabla^{AA'} p) \phi_A\big) \\
&=& e^{iqp} \big(\nabla^{AA'} \phi_A - iq A^{AA'} \phi_A \big) = e^{iqp}\frac{m}{\sqrt{2}} \chi^{A'} = \frac{m}{\sqrt{2}} \tilde{\chi}^{A'} \, .
\end{eqnarray*}
The second equation is treated similarly.

This gauge freedom is particularly useful with regard to the singularity of the electromagnetic potential \eqref{EMPot} at the horizons, due to the presence of the $1$-form $\d x$. The singularity can be removed at one horizon at a time by adding a constant $1$-form to the potential and compensating for it using \eqref{gauge}. Namely, we put
\begin{equation} \label{Gauger-}
\tilde{A} = Q \left( \frac{1}{r} - \frac{1}{r_-} \right) \d x \, ,~ (\tilde{\phi}_A , \tilde{\chi}^{A'} ) = e^{-i\frac{qQ}{r_-}x} (\phi_A , \chi^{A'} )\, .
\end{equation}
This gauge choice gives an electromagnetic potential that is smooth at $\{ r=r_-\}$ because $r-r_-$ decays exponentially fast there. This can be seen explicitely by looking at the metric dual of $\tilde{A}$:
\[ g^{-1} (\tilde{A}) = \frac{Q}{rr_-}\frac{r-r_-}{f} \frac{\partial}{\partial x} \, .\]
Since $\partial_x$ extends smoothly to the horizons as a non-trivial null generator that vanishes at the crossing spheres $\scrs_{r_\pm}$, this vector field is smooth at $\{ r=r_- \}$, but not at the other horizon. The same procedure can be applied for $\{ r=r_+\}$.

\subsection{Conserved current, unitary evolution}

The causal vector field
\begin{equation} \label{Current}
J^a = \phi^A \bar{\phi}^{A'} + \bar{\chi}^A \chi^{A'}
\end{equation}
is a conserved current for Equation \eqref{DirEq}, i.e.
\begin{equation} \label{DivFree}
\nabla^a J_a = 0 \, .
\end{equation}
Indeed, this is an easy consequence of the equation itself:
\begin{eqnarray*}
\nabla^a J_a &=& \nabla^{AA'} \left( \phi_A \bar{\phi}_{A'} + \bar{\chi}_A \chi_{A'} \right) \\
&=& 2 \Re \left( (\nabla^{AA'} \phi_A ) \bar{\phi}_{A'} + (\nabla^{AA'} \chi_{A'} ) \bar{\chi}_A\right) \\
&=& 2 \Re \left( (\nabla^{AA'} \phi_A ) \bar{\phi}_{A'} + (\nabla_{AA'} \chi^{A'} ) \bar{\chi}^A\right) \\
&=& 2 \Re \left( \left( iq A^{AA'} \phi_A + \frac{m}{\sqrt{2}} \chi^{A'} \right) \bar{\phi}_{A'} + \left( i q A_{AA'} \chi^{A'} - \frac{m}{\sqrt{2}} \phi_A \right)  \bar{\chi}^A \right) \\
&=& 2 \Re \left( iq A^{AA'} \phi_A  \bar{\phi}_{A'} +  i q A_{AA'} \chi^{A'} \bar{\chi}^A + \frac{m}{\sqrt{2}} \chi^{A'} \bar{\phi}_{A'} - \frac{m}{\sqrt{2}} \overline{\chi^{A'} \bar{\phi}_{A'}} \right) =0 \, ,
\end{eqnarray*}
since $A^a$ and $m$ are real.

The flux of $J^a$ across a hypersurface $\Sigma_t $ defines the natural $L^2$ norm for $\phi_A \oplus \chi^{A'}$. This can be seen by choosing a normalized Newman-Penrose tetrad $\{ l , n ,m , \bar{m} \}$ on ${\mathcal M}$ that is adapted to the foliation, i.e. such that $l^a + n^a$ is equal to $\sqrt{2}$ times the future-oriented unit normal to $\Sigma_t$, i.e.
\begin{equation} \label{AdaptedTetrad}
l^a+n^a = \sqrt{2} \, \nu^a \, ,~ \nu^a \frac{\partial}{\partial x^a} = \frac{1}{\sqrt{f}}\, \frac{\partial}{\partial t} \, .
\end{equation}
More precisely, we choose
\begin{equation} \label{NPTetrad}
\begin{cases} {l^a \partial_a }&= {\frac{1}{\sqrt{2f}} \big( \frac{\partial}{\partial t} + \frac{\partial}{\partial x} \big) = \sqrt{\frac{2}{f}} \frac{\partial}{\partial v} \, ,}\\
{n^a \partial_a} & = {\frac{1}{\sqrt{2f}} \left( \frac{\partial}{\partial t} - \frac{\partial}{\partial x} \right) = \sqrt{\frac{2}{f}} \frac{\partial}{\partial u} \, ,}\\
{m^a \partial_a} &= {\frac{1}{\sqrt{2}r} \left( \frac{\partial}{\partial \theta} + \frac{i}{\sin \theta} \frac{\partial}{\partial \varphi} \right) \, ,}\\
{\bar{m}^a \partial_a} &= {\frac{1}{\sqrt{2}r} \left( \frac{\partial}{\partial \theta} - \frac{i}{\sin \theta} \frac{\partial}{\partial \varphi} \right) \, .} \end{cases}
\end{equation}
Considering the spin-frame $\{ o^A , \iota^A \}$ related to the tetrad \eqref{NPTetrad} by \eqref{NPTSF}, \eqref{AdaptedTetrad} implies
\begin{equation} \label{FluxSigmat}
\int_{\Sigma_t} J_a \nu^a \mathrm{dVol}_{\Sigma_t} = \frac{1}{\sqrt{2}} \int_{\Sigma_t} \left( \vert \phi_0 \vert^2 + \vert \phi_1 \vert^2 + \vert \chi^{0'} \vert^2 + \vert \chi^{1'} \vert^2 \right) \mathrm{dVol}_{\Sigma_t} \, ,
\end{equation}
where $\mathrm{dVol}_{\Sigma_t}$ is the $3$-volume measure associated to the induced metric on $\Sigma_t$
\begin{equation} \label{dVolSigmat}
\mathrm{dVol}_{\Sigma_t} = \sqrt{f} \,r^2 \d x \d \omega \, ,
\end{equation}
with $\d \omega$ the Lebesgue measure on the euclidean $\s^2$. We define the space ${\cal H}_t$ for each $t\in \R$ by
\begin{equation} \label{SpaceHt}
{\mathcal H}_t := L^2 (\Sigma_t \, ;~ \S_A \oplus \S^{A'} ) \, ,~ \Vert (\phi , \chi ) \Vert^2_{{\mathcal H}_t} = \int_{\Sigma_t} J_a \nu^a \mathrm{dVol}_{\Sigma_t} \, .
\end{equation}
By density, Lemma \ref{LerayCP} together with the divergence theorem, using finite propagation speed and the conservation law \eqref{DivFree}, imply
\begin{lemma}\label{L2CP}
For any $s\in \R$, for any initial data $(\alpha_A \, ,~ \beta^{A'}) \in L^2 (\Sigma_s \, ;~  \S_A \oplus \S^{A'} )$, there exists a unique solution $(\phi_A , \chi^{A'}) \in {\cal C}^0 (\R_t \, ;~ L^2 (\Sigma_t \, ;~  \S_A \oplus \S^{A'} ))$ of \eqref{DirEq} such that
\[ \phi_A \vert_{t=s} = \alpha_A \, ,~ \chi^{A'} \vert_{t=s} = \beta^{A'} \, .\]
Moreover, for all $t \in \R$ we have
\[ \Vert (\phi , \chi ) \Vert_{{\mathcal H}_t} = \Vert (\alpha , \beta ) \Vert_{{\mathcal H}_s} \, .\]
\end{lemma}
\begin{remark}
The conserved current is gauge invariant
\begin{equation} \label{GTCC}
\tilde{J}_a = \tilde{\phi}_A \bar{\tilde{\phi}}_{A'} + \bar{\tilde{\chi}}_{A} \tilde{\chi}_{A'} =J_a \, .
\end{equation}
\end{remark}

\subsection{Explicit expression via the Newman-Penrose formalism}

An explicit expression of \eqref{DirEq} as a symmetric hyperbolic system can be obtained using the Newman-Penrose formalism. The spin-coefficients are the decomposition of the connexion coefficients in the basis given by the Newman-Penrose tetrad.  They are given by the general formulae (see Section 4.5 in Penrose and Rindler Vol. 1 \cite{PeRi})
\begin{gather*}
\kappa = m^a \nabla_l l_a \, ,~\tilde{\rho} = m^a \nabla_{\bar{m}} l_a \,
,~\tilde{\sigma} = m^a \nabla_m l_a \, ,~\tau = m^a \nabla_n l_a \, , \\
\varepsilon = \frac{1}{2} \left( n^a \nabla_l l_a + m^a \nabla_l \bar{m}_a \right)
\, ,~ \alpha = \frac{1}{2} \left( n^a \nabla_{\bar{m}} l_a + m^a \nabla_{\bar{m}}
  \bar{m}_a \right) \\
\beta = \frac{1}{2} \left( n^a \nabla_m l_a + m^a \nabla_m \bar{m}_a
\right) \, ,~ \gamma = \frac{1}{2} \left( n^a \nabla_n l_a + m^a \nabla_n \bar{m}_a
\right) \, , \\
\pi = - \bar{m}^a \nabla_l n^a \, ,~ \lambda = - \bar{m}^a \nabla_{\bar{m}} n^a \,
,~ \mu = - \bar{m}^a \nabla_m n^a \, ,~ \nu = - \bar{m}^a \nabla_n n^a \, .
\end{gather*}
For our choice of tetrad \eqref{NPTetrad}, denoting $f' = \d f / \d r$,
\begin{gather}
\rho = -\mu = \frac{1}{r} \sqrt{\frac{f}{2}}  \, ,~ \gamma = -\varepsilon = \frac{f'}{4\sqrt{2f}} \, ,~ \alpha = -\beta = \frac{-\cot \theta}{2 r \sqrt{2}} \, ,\nonumber \\
\kappa = \sigma = \lambda = \tau = \nu = \pi = 0 \, . \label{SpinCoeffs}
\end{gather}
The charged massive Dirac equation in the Newman-Penrose formalism takes the following form (see for instance \cite{Chan} or \cite{Ha1})
\[
\left\{ \begin{array}{l}
{ n^\mathbf{a} \left( \partial_{\mathbf{a}} - i q A_{\mathbf{a}} \right) \phi_0 - m^\mathbf{a}
 \left( \partial_{\mathbf{a}} - i q A_{\mathbf{a}} \right) \phi_1 + (\mu - \gamma )\phi_0 + (\tau - \beta )
\phi_1 = \frac{m}{\sqrt{2}} \chi^{0'} \, , } \\ \\
{ l^\mathbf{a}  \left( \partial_{\mathbf{a}} - i q A_{\mathbf{a}} \right) \phi_1 - \bar{m}^\mathbf{a}
 \left( \partial_{\mathbf{a}} - i q A_{\mathbf{a}} \right) \phi_0 + (\alpha - \pi )\phi_0 + (\varepsilon -
\rho ) \phi_1 = \frac{m}{\sqrt{2}} \chi^{1'} \, , } \\ \\
{ l^\mathbf{a}  \left( \partial_{\mathbf{a}} - i q A_{\mathbf{a}} \right) \chi^{0'} + m^\mathbf{a}
 \left( \partial_{\mathbf{a}} - i q A_{\mathbf{a}} \right) \chi^{1'}
+ (\bar{\varepsilon} - \bar{\rho} ) \chi^{0'} - (\bar{\alpha} - \bar{\pi} )\chi^{1'}= - \frac{m}{\sqrt{2}}
\phi_{0} \, ,} \\ \\ 
{ n^\mathbf{a}  \left( \partial_{\mathbf{a}} - i q A_{\mathbf{a}} \right) \chi^{1'} + \bar{m}^\mathbf{a}
 \left( \partial_{\mathbf{a}} - i q A_{\mathbf{a}} \right) \chi^{0'} - (\bar{\tau} - \bar{\beta} ) \chi^{0'} + (\bar{\mu} - \bar{\gamma} )\chi^{1'}
 = -\frac{m}{\sqrt{2}} \phi_{1} \, .}\end{array} \right.
\]
In our case, we get the following expression
\[
\left\{ \begin{array}{l}
{ \frac{1}{\sqrt{2f}} \left( \partial_{t} - \partial_x + i \frac{qQ}{r} \right) \phi_0 - \frac{1}{r\sqrt{2}}
 \left( \partial_{\theta} + \frac{i}{\sin \theta} \partial_\varphi \right) \phi_1 - \sqrt{\frac{f}{2}} \left( \frac{1}{r} + \frac{f'}{4f}\right) \phi_0 - \frac{\cot \theta}{2r\sqrt{2}}
\phi_1 = \frac{m}{\sqrt{2}} \chi^{0'} \, , } \\ \\
{ \frac{1}{\sqrt{2f}} \left( \partial_{t} + \partial_x - i \frac{qQ}{r} \right) \phi_1 - \frac{1}{r\sqrt{2}}
 \left( \partial_{\theta} - \frac{i}{\sin \theta} \partial_\varphi \right) \phi_0 - \frac{\cot \theta}{2r\sqrt{2}} \phi_0 - \sqrt{\frac{f}{2}} \left( \frac{1}{r} + \frac{f'}{4f}\right) \phi_1 = \frac{m}{\sqrt{2}} \chi^{1'} \, , } \\ \\
{ \frac{1}{\sqrt{2f}} \left( \partial_{t} + \partial_x - i \frac{qQ}{r} \right) \chi^{0'} + \frac{1}{r\sqrt{2}}
 \left( \partial_{\theta} + \frac{i}{\sin \theta} \partial_\varphi \right) \chi^{1'}
- \sqrt{\frac{f}{2}} \left( \frac{1}{r} + \frac{f'}{4f}\right) \chi^{0'} + \frac{\cot \theta}{2r\sqrt{2}} \chi^{1'}= - \frac{m}{\sqrt{2}}
\phi_{0} \, ,} \\ \\ 
{ \frac{1}{\sqrt{2f}} \left( \partial_{t} - \partial_x + i \frac{qQ}{r} \right) \chi^{1'} + \frac{1}{r\sqrt{2}}
 \left( \partial_{\theta} - \frac{i}{\sin \theta} \partial_\varphi \right) \chi^{0'} + \frac{\cot \theta}{2r\sqrt{2}} \chi^{0'} - \sqrt{\frac{f}{2}} \left( \frac{1}{r} + \frac{f'}{4f}\right) \chi^{1'}
 = -\frac{m}{\sqrt{2}} \phi_{1} \, .}\end{array} \right.
\]
Putting $\Phi = \, ^\mathrm{t} \left( \phi_0 , \phi_1 , \chi^{0'} , \chi^{1'} \right)$, the system above can be written under the form:
\begin{equation} \label{DirEqUnresc}
\left( \partial_t - \frac{f}{r} - \frac{f'}{4} \right) \Phi- \Gamma^1 \left( \partial_x - i \frac{qQ}{r} \right) \Phi - \frac{f^{\frac{1}{2}}}{r} \left( \Gamma^2 \left( \partial_\theta + \frac{\cot \theta}{2} \right) + \Gamma^3 \frac{i}{\sin \theta} \partial_\varphi \right) \Phi = -im f^{\frac{1}{2}} \gamma^0 \Phi \, ,
\end{equation}
where the Dirac matrices are given by
\[ \gamma^0 = i\left( \begin{array}{ll} 0 & \sigma^0 \\ -\sigma^0 & 0 \end{array} \right) \, , ~ \gamma^1 = i\left( \begin{array}{ll} 0 & \sigma^1 \\ \sigma^1 & 0 \end{array} \right) \, ,~ \gamma^2 = i\left( \begin{array}{ll} 0 & \sigma^2 \\ \sigma^2 & 0 \end{array} \right) \, ,~ \gamma^3 = i\left( \begin{array}{ll} 0 & \sigma^3 \\ \sigma^3 & 0 \end{array} \right) \, , \]
the Pauli matrices being
\[ \sigma^0 = \left( \begin{array}{cc} 1 & 0 \\ 0 & 1 \end{array} \right) \, ,~ \sigma^1 = \left( \begin{array}{cc} 1 & 0 \\ 0 & -1 \end{array} \right) \, ,~\sigma^2 = \left( \begin{array}{cc} 0 & 1 \\ 1 & 0 \end{array} \right) \, ,~\sigma^3 = \left( \begin{array}{cc} 0 & i \\ -i & 0 \end{array} \right) \]
and the matrices $\Gamma^1$, $\Gamma^2$, $\Gamma^3$ are
\[ \Gamma^1 = -\gamma^0 \gamma^1=\left( \begin{array}{cc} \sigma^1 & 0 \\ 0 & -\sigma^1 \end{array} \right) \, ,~ \Gamma^2 = -\gamma^0 \gamma^2=\left( \begin{array}{cc} \sigma^2 & 0 \\ 0 & -\sigma^2 \end{array} \right) \, ,~\Gamma^3 = -\gamma^0 \gamma^3=\left( \begin{array}{cc} \sigma^3 & 0 \\ 0 & -\sigma^3 \end{array} \right) \, .\]

\subsection{Evolution system}

Equation \eqref{DirEqUnresc} can be simplified by multiplying the spinor $\Phi$ by the square root of the $3$-volume density of the $t=\mathrm{cst}$ slices. We put
\begin{equation} \label{SpinorRescaling}
\Psi = r f^{\frac{1}{4}} \Phi \, .
\end{equation}
This transformation at each time $t$ is an isometry between ${\cal H}_t$ (described in terms of components of the spinors in the spin-frame $\{ o^A , \iota^A\}$ related by \eqref{NPTSF} to the Newman-Penrose tetrad \eqref{NPTetrad}) and the space
\begin{equation} \label{SpaceCalH}
\mathcal{H} = L^2 (\Sigma \, ;~ \C^4 ) \, ,~ \Vert \Psi \Vert^2_{\mathcal{H}} = \frac{1}{\sqrt{2}} \int_{\Sigma} \vert \Psi \vert^2 \d x \d \omega
\end{equation}
where $\vert \Psi \vert$ is the canonical norm in $\C^4$.

Equation \eqref{DirEqUnresc} then becomes
\begin{equation} \label{DirEqResc1}
\partial_t \Psi- \Gamma^1 \partial_x \Psi - i \frac{f^{\frac{1}{2}}}{r} \sld \Psi + i \Gamma^1 \frac{qQ}{r} \Psi = -im f^{\frac{1}{2}} \gamma^0 \Psi \, ,
\end{equation}
where
\begin{eqnarray*}
\sld &=& -i \left( \Gamma^2 \left( \partial_\theta + \frac{\cot \theta}{2} \right) + \Gamma^3 \frac{i}{\sin \theta} \partial_\varphi \right) = \left( \begin{array}{cc} \sld_{\s^2} & 0 \\ 0 & -\sld_{\s^2} \end{array} \right) \, ,\\
\sld_{\s^2} &=& -i \left( \sigma^2 \left( \partial_\theta + \frac{\cot \theta}{2} \right) + \sigma^3 \frac{i}{\sin \theta} \partial_\varphi \right) \, ,
\end{eqnarray*}
$\sld_{\s^2}$ being the Dirac operator on $\s^2$. 
This equation can be put in Hamiltonian form
\begin{equation} \label{DirEqHamil}
\partial_t \Psi = iH(t) \Psi \, ,~ H(t) = \Gamma^1 D_x + \frac{f^{\frac{1}{2}}}{r} \sld - \Gamma^1 \frac{qQ}{r} - m f^{\frac{1}{2}} \gamma^0\, .
\end{equation}
Since for $i,j\in \{ 1,2,3\}$, $i\neq j$, $\Gamma^i$ and $\Gamma^j$ anti-commute and $\gamma^0$ also anti-commutes with $\Gamma^i$, we have for $q =0$
\[ H(t) ^2 = -\partial_x^2 + \frac{f}{r^2} \sld^2 + f m^2 \, .\]
It follows that the natural domain of $H(t)$ (for any $q \in \R$ and any $m \geq 0$) is
\[ D (H(t)) = \{ \Psi \in \mathcal{H} \, ;~ H(t) \Psi \in \mathcal{H} \} = H^1 (\R \times \s^2 \, ;~ \C^4) \, .\]
The graph norm on the domain of $H(t)$
\[ \Vert \Psi \Vert^2_{D(H(t))} = \Vert \Psi \Vert^2_{\mathcal{H}} + \Vert H(t) \Psi \Vert^2_{\mathcal{H}} \]
is such that any two norms $\Vert . \Vert^2_{D(H(t))}$ and $\Vert . \Vert^2_{D(H(s))}$ are equivalent, but this equivalence is only locally uniform in $(t,s)$ because of the factor $f$ in front of the operator $\sld$. Lemma \ref{L2CP} as well as the above discussion yield
\begin{lemma} \label{HCP}
For any $s\in \R$, for any initial data $\Xi \in \mathcal{H}$, there exists a unique $\Psi \in {\cal C}^0 (\R_t \, ;~ \mathcal{H} )$ solution of \eqref{DirEqHamil} such that $\Psi \vert_{t=s} = \Xi $. Moreover, for all $t \in \R$ we have
\[ \Vert \Psi (t) \Vert_{{\mathcal H}} = \Vert \Xi \Vert_{{\mathcal H}} \, .\]
If in addition $\Xi \in D(H(s))$, then
\[ \Psi \in {\cal C}^0 (\R_t \, ;~ H^1 (\R \times \s^2 \, ;~ \C^4 ) ) \cap {\cal C}^1 (\R_t \, ;~\mathcal{H}) \, .\]
\end{lemma}
We define the evolution system for Equation \eqref{DirEqHamil} as follows.
\begin{definition} \label{Propagator}
The evolution system for Equation \eqref{DirEqHamil} is the family $\{ \mathcal{U} (t,s)\}_{(s,t)\in\R^2}$ of bounded operators on $\mathcal{H}$ such that for any $s\in \R$, for any initial data $\Xi \in \mathcal{H}$, $t \rightarrow \mathcal{U} (t,s) \Xi$ is the solution $\Psi$ given in Lemma \ref{HCP}. It satisfies:
\begin{description}
\item[(1)] for any $(s,t)\in \R^2$, $\mathcal{U}(t,s)$ is a unitary operator on $\mathcal{H}$ and a (locally uniformly in $(t,s)$) bounded operator on $H^1 (\R \times S^2 \, ;~ \C^4)$;
\item[(2)] for any $s\in \R$, $t \mapsto \mathcal{U}(t,s)$ is strongly continuous on $\mathcal{H}$ and on $H^1 (\R \times S^2 \, ;~ \C^4)$;
\item[(3)] for any $t_1, t_2,t_3 \in \R$, $\mathcal{U}(t_3,t_2) \mathcal{U}(t_2,t_1) = \mathcal{U}(t_3,t_1)$ and for any $t \in \R$ $\mathcal{U}(t,t) =\mathrm{Id}_{\mathcal{H}}$;
\item[(4)] for $\Xi \in D(H(s))$,
\[ \frac{\d}{\d t} \mathcal{U}(t,s) \Xi = i H(t) \mathcal{U}(t,s) \Xi\, .\]
\end{description}
These four properties entail the next two:
\begin{description}
\item[(2')] for any $t\in \R$, $s \mapsto \mathcal{U}(t,s)$ is strongly continuous on $\mathcal{H}$ and on $H^1 (\R \times S^2 \, ;~ \C^4)$;
\item[(4')] for $\Xi \in D(H(s))$,
\[ \frac{\d}{\d s} \mathcal{U}(t,s) \Xi = -i \mathcal{U}(t,s) H(s) \Xi\, .\]
\end{description}
\end{definition}
Taking advantage of the spherical symmetry of the spacetime, we can diagonalize the operator $\sld$ using spin-weighted spherical harmonics. For any $m$ such that $2m \in \Z$, let
\[ \mathcal{I}_m=\{ l,n\,;~ l- \vert m\vert \in \mathbb{N}\, ,~l- \vert n\vert \in \mathbb{N} \}\, .\]
The spin-weighted spherical harmonics of spin-weight $m$ are a Hilbert basis $\{ W^{l}_{mn} \}_{(l,n)\in {\mathcal I}_m}$ of $L^2 (\s^2)$. For details of their definition and properties, see for example \cite{Ni1997}. They give us a decomposition of $\Psi \in \mathcal{H}$ according to the spin-weights of its components as follows
\[ \Psi=\sum_{(l,n)\in \I}{\psi}^l_n \odot {W}^l_n \, , ~
{W}^l_n= \, ^\mathrm{t} \left(
W^l_{-\frac{1}{2} n},
W^l_{\frac{1}{2} n},
W^l_{-\frac{1}{2} n},
W^l_{\frac{1}{2} n}
\right)\, , ~ {\psi}^l_n \in \HB= L^2 (\R_x \, ;~ \C^4 ) \]
where $\odot$ denotes the Hadamard product
\[ \left(\begin{array}{c}
x_1\\
y_1\\
z_1\\
t_1
\end{array}\right) \odot \left(\begin{array}{c}
x_2\\
y_2\\
z_2 \\
t_2
\end{array}\right)= \left(\begin{array}{c}
x_1x_2\\
y_1y_2\\
z_1z_2 \\
t_1t_2
\end{array}\right) \, ,\]
such that
\begin{equation}
\sld\Psi=- \Gamma^2 \sum_{(l,n)\in \I}(l+\frac{1}{2}) {\psi}^l_n \odot {W}^l_n \, .
\end{equation}
The hamiltonian $H(t)$ decomposes accordingly
\[H(t)\Psi=\sum_{(l,n) \in \I}H^l(t){\psi}^l_n \odot {W}^l_n \, ,\]
where 
\begin{equation}
H^l(t)=\Gamma^1D_x  -\left(l+\frac{1}{2}\right)\frac{f^\frac{1}{2}}{r}\Gamma^2 -mf^\frac{1}{2} \gamma^0  -\Gamma^1 \frac{qQ}{r},
\end{equation}
whose natural domain is
\[ D (H^l (t) ) = \{ \psi \in \HB \, ;~ H^l (t) \psi \in \HB \} = H^1(\R_x \, ;~ \C^4 ) \, .\]
Note that the graph norm on $D(H^l (t) )$ is uniformly equivalent to the standard $H^1$ norm.

For a fixed angular momentum, i.e. for
\[ \Psi = \psi \odot {W}^l_n \, ,\]
Equation \eqref{DirEqHamil} reduces to
\begin{equation}\label{DiraceqHamiltonianfixmode}
\partial_t\psi=iH^l(t)\psi \, .
\end{equation}

\section{Scattering Theory} \label{ScatteringTheory}

The charge interaction term $\Gamma^1 qQ/r$ has distinct non vanishing limits at both horizons. It is therefore natural to introduce different comparison dynamics at each horizon. Let us define the hamiltonians
\begin{equation} \label{CompHamil}
H_0^\pm = \Gamma^1 D_x - \Gamma^1 \frac{qQ}{r_\mp} \, .
\end{equation}
Both $H^+_0$ and $H_0^-$ are self-adjoint on $\mathcal{H}$ with domain $H^1 (\R_x \, ;~ L^2 (\s^2 \, ;~ \C^4))$. Our main result is the following.
\begin{theorem} \label{ThmScattering}
The future and past direct wave operators $W^\pm$ and inverse wave operators $\Omega^\pm$ are well-defined on $\mathcal{H}$ as the strong limits:
\begin{eqnarray}
	W^\pm &=& s-\lim\limits_{t\rightarrow \pm \infty}\mathcal{U}(0,t)e^{itH_0^\pm},\label{Waveopslim}\\
	\Omega^\pm &=& s-\lim\limits_{t\rightarrow \pm \infty}e^{-itH_0^\pm} \mathcal{U}(t,0)\, .	\label{InvWaveopslim}
\end{eqnarray}
They are unitary operators on $\mathcal{H}$. Moreover,
\begin{equation}
	W^\pm\Omega^\pm=\Omega^\pm W^\pm=\mathrm{Id}_{\mathcal{H}} \, .
	\end{equation}
\end{theorem}
\begin{proof}
We merely need to prove the existence of the limits (\ref{Waveopslim}) and (\ref{InvWaveopslim}), the other properties are then immediate. We start by establishing these limits on the dense subspace of $\mathcal{H}$ generated by elements of the form $\psi \odot W^l_{n}$ where $l \in \I$ and $\psi\in \mathcal{C}^\infty_0 (\R \, ;~\C^4 )$. We perform the proof for the future wave operator $W^+$.

Let $l \in \I$ and $\psi\in \mathcal{C}^\infty_0 (\R \, ;~\C^4 )$, we consider $\Psi_0 = \psi \odot W^l_{n} \in \mathcal{H}$. By Cook's method, it is sufficient to establish that
\[ \partial_t(\mathcal{U}(0,t)e^{itH_0^+}\Psi_0)\in L^1\left(\R^+_t;\HB\right) \, .\]
Using the properties of $\mathcal{U}$ given in Definition \ref{Propagator} ({\it (4')} in particular), we have
\begin{eqnarray*}
\partial_t(\mathcal{U}(0,t)e^{itH_0^+}\Psi_0) &=& \mathcal{U}(0,t)\left( -i H^l + i H_0^+ \right) e^{itH_0^+}\Psi_0 \\
&=& \mathcal{U}(0,t)\left( i\left(l+\frac{1}{2}\right)\frac{f^\frac{1}{2}}{r}\Gamma^2 +imf^\frac{1}{2} \gamma^0  +i \Gamma^1 qQ\left( \frac{1}{r} - \frac{1}{r_-} \right) \right) e^{itH_0^+} \Psi_0 \, .
\end{eqnarray*}
As $t \rightarrow +\infty$,
\[ \left\Vert i\left(l+\frac{1}{2}\right)\frac{f^\frac{1}{2}}{r}\Gamma^2 +imf^\frac{1}{2} \gamma^0  +i \Gamma^1 qQ\left( \frac{1}{r} - \frac{1}{r_-} \right)\right\Vert_{L^\infty (\Sigma_t )} = O (\sqrt{r-r_-}) = O (e^{\kappa_- t}) \, ,\]
therefore, using the unitarity of both $\mathcal{U} (0,t)$ and $e^{itH_0^+}$,
\begin{align*}
\left\Vert \partial_t(\mathcal{U}(0,t)e^{itH_0^+}\Psi_0) \right\Vert_{\mathcal{H}}& = \left\Vert \mathcal{U}(0,t)\left( i\left(l+\frac{1}{2}\right)\frac{f^\frac{1}{2}}{r}\Gamma^2 +imf^\frac{1}{2} \gamma^0  +i \Gamma^1 qQ\left( \frac{1}{r} - \frac{1}{r_-} \right) \right) e^{itH_0^+} \Psi_0 \right\Vert_{\mathcal{H}}\\
&=\left\Vert \left( i\left(l+\frac{1}{2}\right)\frac{f^\frac{1}{2}}{r}\Gamma^2 +imf^\frac{1}{2} \gamma^0  +i \Gamma^1 qQ\left( \frac{1}{r} - \frac{1}{r_-} \right) \right) e^{itH_0^+} \Psi_0 \right\Vert_{\mathcal{H}}\\
&\lesssim e^{\kappa_-t} \left\Vert e^{itH_0^+} \Psi_0 \right\Vert_{\mathcal{H}} = e^{\kappa_-t} \left\Vert \Psi_0 \right\Vert_{\mathcal{H}} \in L^1 ([0,+\infty[_t) \, .
\end{align*}
It follows that the operator $W^+$ is well-defined on a dense subspace of $\mathcal{H}$. Since it is a limit of unitary operators, it satisfies for all $\Psi_0 = \psi \odot W^l_{n}$, $l \in \I$ and $\psi\in \mathcal{C}^\infty_0 (\R \, ;~\C^4 )$,
\[ \Vert W^+ \Psi_0 \Vert_{\mathcal{H}} = \Vert \Psi_0 \Vert_{\mathcal{H}} \, . \]
Therefore $W^+$ extends uniquely to a unitary operator on $\mathcal{H}$, still defined by the strong limit \eqref{Waveopslim}. The proof is analogous for $W^-$ and $\Omega^\pm$.
\end{proof}
\begin{definition} \label{ScatOp}
The scattering operator sums up the evolution of the field from its past scattering data to its future scattering data:
\begin{equation}\label{Scatteringop}
S=\Omega^+W^- .
\end{equation}
It is a unitary operator on $\mathcal{H}$.
\end{definition}
\begin{remark}
Since the Hamiltonian $H(t)$ varies with time, we do not expect entertwining relations to hold between the wave operators and the Hamiltonians. However, the Killing vector fields $\partial_x$ and the generators of rotations all commute with $H(t)$ and $H_0^\pm$. It follows that regularity is transferred either way between the Cauchy and scattering data.
\end{remark}
\begin{remark}[Gauge freedom]
As we will see in Section \ref{GeometricalInterpretation}, when considering the limit $t\rightarrow+\infty$, it is more natural to work with the potential $\tilde{A}$, which is smooth at $r_-$, and therefore to consider the evolution for $\tilde{\Psi}=e^{-i\frac{qQ}{r_-}x}\Psi$. The corresponding hamiltonian is 
\[\tilde{H}(t)=e^{-i\frac{qQ}{r_-}x}H(t)e^{i\frac{qQ}{r_-}x}=\Gamma^1 D_x + \frac{f^{\frac{1}{2}}}{r} \sld - \Gamma^1 \left(\frac{qQ}{r}-\frac{qQ}{r_-}\right) - m f^{\frac{1}{2}} \gamma^0 \, .\]
Let $\tilde{\mathcal U}(t,s)=e^{-i\frac{qQ}{r_-}x}{\mathcal U}(t,s)e^{i\frac{qQ}{r_-}x}$ be the corresponding evolution system and $\tilde{H}_0^+=\Gamma^1 D_x$. From the previous results we obtain immediately the existence of the wave operators 
\begin{eqnarray*}
\tilde{\Omega}^+&=&s-\lim_{t\rightarrow+\infty}e^{-it\tilde{H}_0^+}\tilde{\mathcal U}(t,0)=e^{-i\frac{qQ}{r_-}x}\Omega^+e^{i\frac{qQ}{r_-}x},\\
\tilde{W}^+&=&s-\lim_{t\rightarrow+\infty}\tilde{\mathcal U}(0,t)e^{it\tilde{H}_0^+}=e^{-i\frac{qQ}{r_-}x}W^+e^{i\frac{qQ}{r_-}x}. 
\end{eqnarray*}
A similar remark holds for the limit $t\rightarrow-\infty$, for which it is natural to work with the potential $\hat{A}=Q\left(\frac{1}{r}-\frac{1}{r_+}\right) \d x$. 
\end{remark}

\section{Geometric Interpretation} \label{GeometricalInterpretation}

In this section we will show how the inverse wave operators can be understood as trace operators on the horizons and the direct wave operators as providing the solution to Goursat problems set at the horizons. We detail this interpretation for $t\rightarrow +\infty$, i.e. at the Cauchy horizon. The construction as $t\rightarrow -\infty$ is analogous. We start by defining natural trace operators using essentially the gauge freedom of the equation and the theorems by Leray on hyperbolic equations. Then we interpret the wave operators, after some phase modifications, as a realisation of these trace operators in a different spin dyad.

\subsection{Trace operators}

First, let us use the gauge transformation \eqref{Gauger-}: for $(\phi_A , \chi^{A'})$ solution to \eqref{DirEq}, the Dirac spinor
\[ \tilde{\Upsilon} := (\tilde{\phi}_A , \tilde{\chi}^{A'}) = e^{-i\frac{qQ}{r_-}x} (\phi_A , \chi^{A'}) \]
satisfies Equation \eqref{GTDirEq} with the electromagnetic potential
\[ \tilde{A} = Q \left( \frac{1}{r} - \frac{1}{r_-} \right) \d x \, .\]
This equation is hyperbolic with smooth coefficients on $\bar{\mathcal{M}}^+ = \mathcal{M} \cup \{r=r_-\} $. Therefore we can use Leray's theorems to infer the following result.
\begin{proposition} \label{Leray}
Let $(\tilde{\alpha}_A , \tilde{\beta}^{A'} ) \in \mathcal{C}^\infty_0 (\Sigma_0 \, ;~ \S_A \oplus \S^{A'} )$ and let $(\tilde{\phi}_A , \tilde{\chi}^{A'}) \in \mathcal{C}^\infty (\mathcal{M}\,;~\S_A \oplus \S^{A'} )$ be the solution to \eqref{GTDirEq} whose restriction on $\Sigma_0$ is equal to $(\tilde{\alpha}_A , \tilde{\beta}^{A'} )$. Then $(\tilde{\phi}_A , \tilde{\chi}^{A'})$ extends as a smooth section of $\S_A \oplus \S^{A'}$ on $\bar{\mathcal{M}}^+$. The conserved current \eqref{GTCC} for \eqref{GTDirEq} means that
\begin{equation} \label{CCBar}
\Vert (\tilde{\alpha}_A , \tilde{\beta}^{A'} ) \Vert_{\mathcal{H}}^2 = \int_{{\hlm} \cup {\hrm}} \tilde{J}_a \nu^a \d \sigma \, ,
\end{equation}
where $\nu^a$ is a future-oriented null generator of $\hlm$ or $\hrm$ and $\d \sigma = V \lrcorner \d^4 \mathrm{Vol}_g$ where $V$ is a vector field transverse to $\hlm$ or $\hrm$ such that $\nu_aV^a=1$.
\end{proposition}
We can choose a spin-frame that extends regularly to $\hlm$ or $\hrm$ to make \eqref{CCBar} more explicit. We consider the two normalized Newman-Penrose tetrads expressed in coordinates $(u,v,\theta,\varphi)$ and in relation to the tetrad \eqref{NPTetrad},
\[
\begin{cases} {_{_1}l^a \partial_a }&= {r^2\sqrt{2} \frac{\partial}{\partial v} = r^2\sqrt{f} l^a \partial_a\, ,}\\
{_{_1}n^a \partial_a} & =  {\frac{\sqrt{2}}{r^2f} \frac{\partial}{\partial u} = \frac{1}{r^2\sqrt{f}}n^a\partial_a\, ,}\\
{_{_1}m^a \partial_a} &= {\frac{1}{r\sqrt{2}} \left( \frac{\partial}{\partial \theta} + \frac{i}{\sin \theta} \frac{\partial}{\partial \varphi} \right) = m^a \partial_a\, ,}\\
{_{_1}\bar{m}^a \partial_a} &= {\frac{1}{r\sqrt{2}} \left( \frac{\partial}{\partial \theta} - \frac{i}{\sin \theta} \frac{\partial}{\partial \varphi} \right) = \bar{m}^a \partial_a\, ,} \end{cases} 
\hspace{0.5in} 
\begin{cases} {_{_2}l^a \partial_a }&= {\frac{\sqrt{2}}{r^2f} \frac{\partial}{\partial v} =\frac{1}{r^2\sqrt{f}}l^a\partial_a \, ,}\\
{_{_2}n^a \partial_a} & =  {r^2\sqrt{2} \frac{\partial}{\partial u} =r^2\sqrt{f}n^a\partial_a\, ,}\\
{_{_2}m^a \partial_a} &= {\frac{1}{r\sqrt{2}} \left( \frac{\partial}{\partial \theta} + \frac{i}{\sin \theta} \frac{\partial}{\partial \varphi} \right) = m^a \partial_a\, ,}\\
{_{_2}\bar{m}^a \partial_a} &= {\frac{1}{r\sqrt{2}} \left( \frac{\partial}{\partial \theta} - \frac{i}{\sin \theta} \frac{\partial}{\partial \varphi} \right) = \bar{m}^a \partial_a} \end{cases} 
\]
and the associated spin-frames $\{ \,_{_1}o^A \, ,~_{_1}\iota^A \}$ and $\{ \,_{_2}o^A \, ,~_{_2}\iota^A \}$, related to the spin-frame associated to \eqref{NPTetrad} by
\[ \,_{_1}o^A = rf^{1/4} o^A\, ,~_{_1}\iota^A = \frac{1}{rf^{1/4}} \iota^A \, ,~ \,_{_2}o^A = \frac{1}{rf^{1/4}}  o^A\, ,~_{_2}\iota^A = rf^{1/4} \iota^A \, . \] The first one extends smoothly at $\hlm$ and the second one at $\hrm$ and we have
\begin{eqnarray}
\int_{{\hlm} } \tilde{J}_a \nu^a \d \sigma &=& \int_{\R_v \times \s^2_\omega} \,_{_1}l^a \tilde{J}_a \d v \d^2 \omega =\frac{1}{\sqrt{2}}\int_{\R_v \times \s^2_\omega} \left( \vert \tilde{\phi}_0 \vert^2 + \vert \tilde{\chi}^{1'} \vert^2  \right) \d v \d^2 \omega \, , \nonumber \\
\int_{{\hrm} } \tilde{J}_a \nu^a \d \sigma &=& \int_{\R_u \times \s^2_\omega} \, _{_2}n^a \tilde{J}_a \d u \d^2 \omega = \frac{1}{\sqrt{2}}\int_{\R_u \times \s^2_\omega} \left( \vert \tilde{\phi}_1 \vert^2 + \vert \tilde{\chi}^{0'} \vert^2  \right) \d u \d^2 \omega \, , \nonumber 
\end{eqnarray}
where at each horizon, the components of the spinors are taken with respect to the spin-frame that extends smoothly there. Note that $(\tilde{\phi}_0,\, \tilde{\chi}^{1'})$ at $\hlm$ are thus the first and fourth components of the spinor $rf^{1/4}\tilde{\Upsilon}$ in the tetrad  \eqref{NPTetrad}, an equivalent statement holds for $(\tilde{\phi}_1,  \tilde{\chi}^{0'})$ at $\hrm$. We now define the following trace operators.
\begin{definition}
With the same notations as above, we define the operators
\begin{gather}
T^+_{L} \, :~ {\cal C}^\infty_0 (\Sigma_0 \, ;~ \S_A\oplus \S^{A'}) \longrightarrow {\cal C}^\infty (\hlm \, ;~ \C^2 ) \, ,~ T^+_{L} (\tilde{\alpha}_A , \tilde{\beta}^{A'}) = (\tilde{\phi}_0,\tilde{\chi}^{1'}) \vert_{\hlm} \, ,\label{Tracehlm} \\
T^+_{R} \, :~ {\cal C}^\infty_0 (\Sigma_0 \, ;~ \S_A\oplus \S^{A'}) \longrightarrow {\cal C}^\infty (\hrm \, ;~ \C^2 ) \, ,~ T^+_{R} (\tilde{\alpha}_A , \tilde{\beta}^{A'}) = (\tilde{\phi}_1,\tilde{\chi}^{0'})\vert_{\hrm} \, .\label{Tracehrm} 
\end{gather}
\end{definition}
Proposition \ref{Leray} entails the following result.
\begin{proposition} \label{ThmTO}
The trace operator $T^+_{L}$ (resp. $T^+_{R}$) extends as a partial isometry (i.e. norm-preserving but not necessarily surjective) from $\mathcal{H}_0$ into $L^2 (\hlm \, ;~ \C^2)$ (resp. $L^2 (\hrm \, ;~ \C^2)$). 
\end{proposition}
\begin{remark}
A similar construction can be performed at the bifurcate horizon $\{ r=r_+ \}$, but the relevant components at the left-hand and right-hand sides of the horizon are reversed from the Cauchy horizon.
\end{remark}
The main result of this section is the following :
\begin{theorem}
\label{thm5.2}
The combined trace operator
\[ T^+:=(T^+_{\mathrm{L}} \, ,~ T^+_{\mathrm{R}}) \, :~ \mathcal{H}_0 \longrightarrow L^2 (\hlm \, ;~ \C^2) \oplus L^2 (\hrm \, ;~ \C^2)\]
is an isometry (i.e. a norm-preserving isomorphism). In addition, the trace operators preserve tangent regularity; more precisely, we have for all $M,N\in \N$ and for all $\tilde{\Upsilon} \in \mathcal{H}_0$:
\begin{eqnarray}
\label{TR1}
\partial_v^M(-\Delta_{\omega})^{N/2}T^+_L\tilde{\Upsilon}=T^+_L\partial_x^M(-\Delta_{\omega})^{N/2}\tilde{\Upsilon},\\
\label{TR2}
\partial_u^M(-\Delta_{\omega})^{N/2}T^+_R\tilde{\Upsilon}=-T^+_R\partial_x^M(-\Delta_{\omega})^{N/2}\tilde{\Upsilon}.
\end{eqnarray}
\end{theorem}
\begin{remark}
A similar result is valid at the event horizon $\{ r=r_+ \}$ giving rise to a trace operator $T^-$. We can then build the scattering matrix ${\mathcal S}=T^+{\mathcal G}(T^-)^{-1}$ which is an isometry between 
$L^2 (\hlp \, ;~ \C^2) \oplus L^2 (\hrp \, ;~ \C^2)$ and $L^2 (\hlm \, ;~ \C^2) \oplus L^2 (\hrm \, ;~ \C^2)$. Here ${\mathcal G}=e^{-i\frac{qQ}{r_-}x}e^{i\frac{qQ}{r_+}x}$ is a gauge changing operator.  
\end{remark}

\subsection{Relation to the wave operators}

We first define the unitary operator 
\begin{equation*}
{\mathcal B}: \begin{array}{rcl} {\mathcal H}_0&\rightarrow&\mathcal{H},\\ (\phi , \chi)&\mapsto&(\phi_0,\phi_1,\chi^{0'},\chi^{1'}). \end{array} 
\end{equation*}
Here the components are taken in spin-frame associated to the tetrad \eqref{NPTetrad}. There are natural diffeomorphisms between $\hlm$ resp. $\hrm$ and $\Sigma_0$ by identifying points along incoming resp. outgoing radial null geodesics. Concretely we obtain
\begin{eqnarray*}
{\mathcal I}_L&:&\begin{array}{rcl} \hlm & \rightarrow & \Sigma_0,\\ (v,\omega) & \mapsto & (x=v,\omega). \end{array}\\
{\mathcal I}_R&:&\begin{array}{rcl} \hrm & \rightarrow & \Sigma_0,\\ (u,\omega) & \mapsto & (x=-u,\omega). \end{array}
 \end{eqnarray*}
We will also need the projections $P_{ij}$ and the embeddings $E_{ij}$:
\begin{eqnarray*}
P_{ij}&:&\begin{array}{rcl} \C^4&\rightarrow&\C^2,\\(z_1,z_2,z_3,z_4)&\mapsto&(z_i,z_j),\end{array}\\
E_{14}&:&\begin{array}{rcl} \C^2&\rightarrow&\C^4,\\ (z_1,z_2)&\mapsto&(z_1,0,0,z_2),\end{array}\\
E_{23}&:&\begin{array}{rcl} \C^2&\rightarrow&\C^4,\\ (z_1,z_2)&\mapsto&(0,z_1,z_2,0).\end{array}
\end{eqnarray*}
%
%
\begin{proposition}
We have for $\tilde{\Upsilon}\in \mathcal{H}_0$
\label{prop5.2}
\begin{eqnarray}
\label{T_L}
T^+_L\tilde{\Upsilon}&=&{\mathcal I}^*_LP_{1,4}\tilde{\Omega}^+rf^{1/4}{\mathcal B}\tilde{\Upsilon},\\ 
\label{T_R}
T^+_R\tilde{\Upsilon}&=&{\mathcal I}^*_R P_{2,3}\tilde{\Omega}^+rf^{1/4}{\mathcal B}\tilde{\Upsilon},
\end{eqnarray}
where the $*$ denotes the push-forward map.
\end{proposition}
\begin{remark}
The proposition gives in particular another proof that the trace operators extend to bounded operators between $L^2$ spaces. 
\end{remark}
\begin{proof}
We only prove \eqref{T_L}; the proof of \eqref{T_R} being analogous. It is sufficient to establish the result for $\tilde\Upsilon\in C_0^{\infty}(\Sigma_0\, ;~\S_A \oplus \S^{A'} )$. Let us consider $P_{1,4}\tilde{\Omega}^+rf^{1/4}{\mathcal B}\tilde{\Upsilon}$. This equals 
\begin{equation*}
\lim_{t\rightarrow\infty}P_{1,4}e^{-it\tilde{H}_0^+}\tilde{\mathcal U}(t,0)rf^{1/4}{\mathcal B}\tilde{\Upsilon}.
\end{equation*}
The function
\[\Psi(t)=\tilde{\mathcal U}(t,0){rf^{1/4}\mathcal B}\tilde{\Upsilon}\]
is the solution at time $t$ to the Dirac equation with potential $\tilde{A}$ and the initial data $\tilde{\Upsilon}$, rescaled by $rf^{1/4}$ and projected onto the spin-frame associated to the tetrad \eqref{NPTetrad}. 
\[P_{1,4}e^{-it\tilde{H}^+_0}\Psi(t)=(\Psi_1(t,x-t,\omega),\Psi_4(t,x-t,\omega))\] 
are the first and fourth components of this solution at time $t$, pulled back to $\Sigma_0$ along the flow of incoming radial null geodesics; it is an element of $L^2((\Sigma_0,\d x \d^2\omega);\C^2)$. As $t\rightarrow+\infty$, this tends to
$T^+_L(\tilde{\Upsilon})({\mathcal I}_L^{-1}(x,\omega)).$ 
\end{proof}
\subsection{Proof of Theorem \ref{thm5.2}}
It is sufficient to construct a right inverse for the trace operator $T^+$. We put 
\begin{gather*}
\hat{T}^+ \, :~ L^2 (\hlm \, ;~ \C^2) \oplus L^2 (\hrm \, ;~ \C^2) \longrightarrow \mathcal{H}_0 \, ,\\
\hat{T}^+:={\mathcal B}^{-1}r^{-1}f^{-1/4}\tilde{W}^+E_{1,4}({\mathcal I}_L^*)^{-1}+{\mathcal B}^{-1}r^{-1}f^{-1/4}\tilde{W}^+E_{2,3}({\mathcal I}_R^*)^{-1}.
\end{gather*} 
An easy calculation using Proposition \ref{prop5.2} then gives 
\begin{equation*}
T^+\circ\hat{T}^+=({\mathcal I}^*_LP_{1,4}E_{1,4}({\mathcal I}_L^*)^{-1}+{\mathcal I}^*_LP_{1,4}E_{2,3}({\mathcal I}_R^*)^{-1},{\mathcal I}^*_RP_{2,3}E_{1,4}({\mathcal I}_L^*)^{-1}+{\mathcal I}^*_RP_{2,3}E_{2,3}{(\mathcal I}_R^*)^{-1})=I. 
\end{equation*}
It remains to show \eqref{TR1} and \eqref{TR2}. We use Proposition \ref{prop5.2}. Let us consider \eqref{TR1}, the proof for \eqref{TR2} being analogous.  
Clearly $\Delta_{\omega}$ commutes with all operators. As for $\partial_v$ we first observe that $\partial_v{\mathcal I}^*_L={\mathcal I}^*_L\partial_x$. We then note that $\partial_x$ commutes with $\tilde{H}_0^+$ and $\tilde{\mathcal U}(t,0)$. 

\qed
 \appendix
 \section{Configurations of horizons in the anti-de Sitter case}
 
 In this short appendix we give the precise conditions on the parameters of $F$ in \eqref{DSRNMet} that determine the number of its real zeros in the case $\Lambda<0$. We use the same method as in M. Mokdad \cite{Mo1} where the case $\Lambda >0$ was treated\footnote{Note that if $F$ has three positive zeros then necessarily $\Lambda>0$.}, however, the analysis for the present case is much simpler. Let $M>0$, we multiply $F$ by $r^2$, and consider $P(r) = -\Lambda r^4 +r^2 -2Mr  +Q^2$ with $\Lambda<0$.
 
 From 
 $P'(r)=-4\Lambda r^3+2r-2M$ and  
 $P''(r)=-12\Lambda r^2 +2$,
 we see that $P''>0$ and hence $P'$ is strictly increasing over $\mathbb{R}$. Let $s$ be the only root of $P'$, and note that $s>0$. There are three cases:
 \begin{itemize}
 	\item $P$ has no roots if and only if $P(s)>0$;
 	\item $P$ has one double root if and only if $P(s)=0$;
 		\item $P$ has two simple roots if and only if $P(s)<0$.
 \end{itemize}

Now we write $P(r)=rP'(r)+T(r)$ and we have $T(r) = 3 \Lambda r^4 - r^2 + Q^2 =: Z(r^2)$. The polynomial $Z$ has positive discriminant and its two roots are
\[ \frac{1+ \sqrt{1-12\Lambda Q^2}}{6\Lambda}<0 < \frac{1- \sqrt{1-12\Lambda Q^2}}{6\Lambda} =:z^2 \, .\]
Hence $T$ has two simple roots $\pm z$ with $z>0$ if $Q\neq 0$, and one double root $z=0$ if $Q=0$. Note that $T$ is negative outside the roots and positive between them. First, assume $Q\neq0$. Since $P(s)=T(s)$, the above three cases become
\begin{itemize}
	\item $P$ has no roots if and only if $0<s<z$ which is equivalent to $P'(z)>0$,
	\item $P$ has one double root if and only if $s=z$ which is equivalent to $P'(z)=0$,
	\item $P$ has two simple roots if and only if $s>z$ which is equivalent to $P'(z)<0$,
\end{itemize}
where we have used the fact that $P'$ is strictly increasing. Note that the roots in the above cases are positive when they exist, as is obvious from the expression of $P$. Finally, if $Q=0$, $T$ is negative everywhere except at zero where it vanishes. In this case, zero is a root of $P$ and there is only one other real root, which is positive since $s>0$. The result relevant to this paper is that for $\Lambda<0$, $F$ defined in \eqref{DSRNMet} has two positive zeros if and only if $Q\ne0$ and $z-2\Lambda z^3<M$.

\end{document}